%% file: main.tex
\documentclass[12pt]{article}
\usepackage{fullpage}
\usepackage{amsmath}
\usepackage{amsfonts}
\usepackage{amssymb}
\usepackage{amsthm}
\usepackage{complexity}
\usepackage{hyperref}
\usepackage{appendix}

\newcommand{\F}{\mathbb{F}}

\DeclareMathOperator{\maxrank}{\mathrm{max-rank}}
\DeclareMathOperator{\rank}{\mathrm{rank}}

\newtheorem{theorem}{Theorem}[section]
\newtheorem{corollary}[theorem]{Corollary}
\newtheorem{lemma}[theorem]{Lemma}
\newtheorem{observation}[theorem]{Observation}
\newtheorem{proposition}[theorem]{\bf Proposition}

\newtheorem*{proposition-a}{\bf Proposition}
\newtheorem*{lemma-a}{\bf Lemma}

\newtheorem{definition}[theorem]{Definition}

\newenvironment{proof-sketch}{\noindent{\bf Sketch of Proof}\hspace*{1em}}{\qed\bigskip}
\newenvironment{proof-idea}{\noindent{\bf Proof Idea}\hspace*{1em}}{\qed\bigskip}
\newenvironment{proof-of-lemma}[1]{\noindent{\bf Proof of Lemma #1}\hspace*{1em}}{\qed\bigskip}
\newenvironment{proof-attempt}{\noindent{\bf Proof Attempt}\hspace*{1em}}{\qed\bigskip}
\newenvironment{proof-of}[1]{\noindent{\bf Proof
of #1:}\hspace*{1em}}{\qed\bigskip}

\title{Arithmetic Circuit Lower Bounds via MaxRank\footnote{Part of this work was done while the first two authors were undergraduate students at IIT Madras}}

\author{Mrinal Kumar\footnote{Department of Computer Science, Rutgers University, Piscataway, NJ 08855, USA. Email : {\tt mrinal.kumar@rutgers.edu}} \and Gaurav Maheshwari\footnote{Department of Computer Science and Engineering, Indian Institute of Technology Madras, Chennai - 600036. Email:{\tt gaurav.m.iitm@gmail.com}} \and Jayalal Sarma M.N.\footnote{Department of Computer Science and Engineering, Indian Institute of Technology Madras, Chennai - 600036. Email :{\tt jayalal@cse.iitm.ac.in}}}

\begin{document}
\maketitle
\begin{abstract}
We introduce the polynomial coefficient matrix and identify maximum rank 
of this matrix under variable substitution as a complexity measure for 
multivariate polynomials. We use our techniques to prove super-polynomial lower bounds against several classes of non-multilinear arithmetic circuits. In particular, we obtain the following results :
\begin{itemize}
\item As our main result, we prove that any homogeneous depth-$3$ circuit for computing the product of $d$ matrices of dimension $n \times n$ requires $\Omega(n^{d-1}/2^d)$ size. This improves the lower bounds in \cite{NW95} when $d=\omega(1)$.
\item There is an explicit polynomial on $n$ variables and degree at most $\frac{n}{2}$ for which any depth-3 circuit $C$ of product dimension at most $\frac{n}{10}$ (dimension of the space of affine forms feeding into each product gate) requires size $2^{\Omega(n)}$. This generalizes the lower bounds against diagonal circuits proved in \cite{S07}. Diagonal circuits are of product dimension $1$.
\item We prove a $n^{\Omega(\log n)}$ lower bound on the size of
  product-sparse formulas. By definition, any multilinear formula
  is a product-sparse formula. Thus, our result extends the known
  super-polynomial lower bounds on the size of multilinear formulas
  \cite{Raz06}.

\item We prove a $2^{\Omega(n)}$ lower bound on the size of partitioned arithmetic branching programs. This result extends the known exponential lower bound on the size of ordered arithmetic branching programs~\cite{J08}.
\end{itemize}
\end{abstract}

\input{introduction.tex}
\input{prelims.tex}
\input{tool.tex}
\input{homogen.tex}
\input{lowrank.tex}

\input{prod-sparse.tex}
\input{part-abp.tex}

\section{Acknowledgements}
The authors thank the anonymous referees for suggesting a simplified view of the proof for Lemma~\ref{l:8}. 
The first author thanks Shubhangi Saraf and Venkata Koppula for some insightful discussions related to Section~\ref{sec:lowrank}.
\bibliography{refs}
\newpage
\appendix

\input{appendix.tex}

\end{document}

%% file: introduction.tex
\section{Introduction} \label{sec:intro}
Arithmetic circuits is a fundamental model of computation for polynomials. 
Establishing the limitations of polynomial sized arithmetic circuits is a central open question in the area of algebraic complexity(see~\cite{SY10} for a detailed survey). 
%
%
One of the surprises in the area was the result due to Agrawal and Vinay~\cite{AV08} where they show that if a polynomial in $n$ variables of degree $d$ (linear in $n$) can be computed by arithmetic circuits of size $2^{o(n)}$, then it can be computed by depth-$4$ circuits of size $2^{o(n)}$. The parameters of this result was further tightened by Koiran~\cite{Koi10}. These results explained the elusiveness of proving lower bounds against even depth-$4$  circuits. For depth-$3$ circuits, the best known general result (over finite fields) is an exponential lower bound due to Grigoriev and Karpinski~\cite{GK98} and Grigoriev and Razborov~\cite{GR98}. Lower bounds for restricted classes of depth-$3$ and depth-$4$ circuits are studied in \cite{ASSS12,NW95,SW01} .

One class of models which has been extensively studied is when the gates are restricted to compute multilinear polynomials. Super-polynomial lower bounds are known for the size of multilinear formulas computing the permanent/determinant polynomial\cite{Raz09}. However, even under this restriction proving super-polynomial lower bounds against arbitrary multilinear arithmetic circuits is an open problem (see \cite{SY10}). The parameter identified by~\cite{Raz06}, which showed the limitations of multilinear formulas, was the rank of a matrix associated with the circuit - namely the partial derivatives matrix\footnote{An exponential sized matrix associated with the multilinear polynomial with respect to a partition of the variables into two sets. See Section~\ref{sec:prelims} for the formal definition.}. The method showed that there exists
a partition of variables into two sets such that the rank of the partial derivatives matrix of any polynomial computed by the model is upper bounded by a function of the size of the circuit. But there are explicit polynomials for which the rank of the partial derivatives matrix is high. 
This program has been carried out for several classes of multilinear polynomials and several variants of multilinear circuits\cite{DMPY11, J08, RY08b, Raz06, Raz09, RSY08}. However, this technique has inherent limitations when it comes to proving lower bounds against non-multilinear circuits because the partial derivatives matrix, in the form that was studied, can be considered only for multilinear circuits.

In this work, we generalize this framework to prove lower bounds against certain classes of non-multilinear arithmetic circuits. This generalization also shows that the multilinearity restriction in the above proof strategy can possibly be eliminated from the circuit model side. Hence it can also be seen as an approach towards proving lower bounds against the general arithmetic circuits.

We introduce a variant of the partial derivatives matrix where the entries will be polynomials instead of constants - which we call the {\em polynomial coefficient matrix}. Instead of rank of the partial derivatives matrix, we analyze the {\em maxrank} - the maximum rank of the polynomial coefficient matrix\footnote{When it is clear from the context, we drop the matrix as well as the partition. By the term, maxrank of a polynomial, we denote the maximum rank of the polynomial coefficient matrix corresponding to the polynomial with respect to the partition in the context.} under any substitution for the variables from the underlying field. We first prove how the maxrank changes under arithmetic operations. 
These tools are combined to prove upper bounds on maxrank of various restrictions of arithmetic circuits.


In \cite{NW95}, it was proved that any homogeneous depth-$3$ circuit for multiplying $d$ $n \times n$ matrices requires $\Omega \left(n^{d-1}/d! \right)$ size. We use our techniques to improve this result in terms of the lower bound. Our methods are completely different from \cite{NW95} and this demonstrates the power of this method beyond the reach of the original partial derivatives matrix method due to Raz~\cite{Raz06}. We prove the following as the main result of this paper.

\begin{theorem}[Main Result]
\label{introthm:hom}
Any homogeneous depth-$3$ circuit for computing the product of $d$ matrices of dimension $n \times n$ requires $\Omega(n^{d-1}/2^d)$ size.
\end{theorem}

Notice that compared to the bounds in \cite{NW95}, our bounds are stronger when $d=\omega(1)$. Very recently, Gupta {et al}\cite{GKKS12} studied the model of homogeneous circuits and proved a strong lower bound parameterized by the bottom fan-in. They studied depth-$4$ circuits ($\Sigma\Pi\Sigma\Pi$) and showed that if the fan-in of the bottom level product gate of the circuits is $t$, then any homogeneous depth-$4$ circuit computing the permanent (and the determinant) must have size $2^{\Omega(\frac{n}{t})}$. In particular, this implies $2^{\Omega(n)}$ lower bound for any depth-$3$ homogeneous circuit computing the permanent (and the determinant) polynomial of $n \times n$ matrices ($n^2$ variables). However, we remark that Theorem~\ref{introthm:hom} is addressing the iterated matrix multiplication polynomial and hence is not directly subsumed by the above result. Moreover, the techniques used in \cite{GKKS12} are substantially different from ours.

We apply our method to depth-$3$ circuits where space of the affine forms feeding into each product gate in the circuit is of limited dimension. Formally,  a depth-$3$ $\Sigma\Pi\Sigma$ circuit $C$ is said to be of product dimension $r$ if for each product gate $P$ in $C$, where $P = \Pi_{i=1}^{d}{L_i}$, where $L_i$ is an affine form for each $i$, the dimension of the span of the set $\{L_i\}_{i \in [d]}$ is at most $r$. 

We prove exponential lower bounds on the size (in fact, the top fan in) of depth-$3$ circuits of bounded product dimension for computing an explicit function.

\begin{theorem}
\label{introthm:lowrank}
There is an explicit polynomial on $n$ variables and degree at most $\frac{n}{2}$ for which any $\Sigma\Pi\Sigma$ circuit $C$ of product dimension at most $\frac{n}{10}$ requires size $2^{\Omega(n)}$.
\end{theorem}

In \cite{S07}, the author studies diagonal circuits, which are depth-$3$ circuits where each product gate is an exponentiation gate. Clearly, such a product gate can be visualized as a product gate with the same affine form being fed into it multiple times. Thus, these circuits are of product dimension $1$, and our lower bound result generalizes size lower bounds against diagonal circuits.

Note that the product dimension of a depth-$3$ circuit is different from the dimension of the span of all affine forms computed at the bottom sum gates of a $\Sigma\Pi\Sigma$ circuit. We will show that this parameter, which we refer to as the total dimension of the circuit, when bounded, makes the model non-universal.




For our next result, we generalize the model of syntactic multilinear formulas to product-sparse formulas (see section~\ref{sec:prelims} for a definition). These formulas can compute non-multilinear polynomials as well. We show the following theorem regarding this model using our methods.
\begin{theorem}
Let $X$ be a set of $2n$ variables and let $f \in \F[X]$ be a full 
max-rank polynomial. Let $\Phi$ be any $(s,d)$-product-sparse formula 
of size $n^{\epsilon \log n}$, for a constant $\epsilon$. If $sd = o(n^{1/8})$, then $f$ cannot be computed by $\Phi$.
\end{theorem}

We also generalize the above theorem to the case of preprocessesed product-sparse formulas. A preprocessed product-sparse formula can be viewed as obtained from a product-sparse formula by applying a {\em preprocessing 
step} in which each occurrence of input variables is replaced by a non-constant univariate polynomial. Different instances of the same input variable are allowed to be replaced by different univariate polynomials.


As our fourth result, we define partitioned arithmetic branching programs which are generalizations of ordered ABPs. While ordered ABP can only compute multilinear polynomials, partitioned ABP is a non-multilinear model, thus can compute non-multilinear polynomials too. Moreover, exponential lower bounds are known for ordered arithmetic branching programs~\cite{J08}.
We prove an exponential lower bound for partitioned ABPs.

\begin{theorem}
 Let $X$ be a set of $2n$ variables and $\F$ be a field. For any full max-rank 
 homogeneous polynomial $f$ of degree $n$ over $X$ and $\F$, the size of any 
 partitioned ABP computing $f$ must be $2^{\Omega(n)}$.
\end{theorem}

The rest of the paper is organized as follows.
In section~\ref{sec:prelims} we describe formally some of the preliminary definitions and notations. In Section~\ref{sec:tool} we define the main parameter of our lower bounds - the polynomial coefficient matrix and prove the required properties with respect to arithmetic operations. Section~\ref{sec:depth3} presents the lower bound result against depth-$3$ homogeneous circuits for computing iterated matrix multiplication. In section~\ref{sec:lowrank}, we present an exponential lower bound for $\Sigma\Pi\Sigma$ circuits of bounded product dimension.  In section~\ref{sec:product-sparse}, we present super-polynomial lower bounds on preprocessed product-sparse formulas. In section~\ref{sec:partabp} we prove exponential lower bounds on partitioned arithmetic branching programs.

%% file: prelims.tex
\section{Preliminaries}
\label{sec:prelims}
In this section, we formally define the models we study. For more detailed account of the model and the results we refer the reader to the survey~\cite{SY10}. 

 An arithmetic circuit $\Phi$ over the field $\F$ and the set of variables 
 $X = \{x_1, x_2, \dots , x_n \}$ is a directed acyclic graph $G = (V,E)$. 
 The vertices of $G$ with in-degree $0$ are called {\em input} gates and 
 are labelled by variables in $X$ or constants from the field $\F$. The 
 vertices of $G$ with out-degree $0$ are called {\em output} gates. Every 
 internal vertex is either a plus gate or a product gate. 
 We will be working with arithmetic circuits with a single output 
 gate and fan-in of every vertex being at most two. The polynomial computed by the arithmetic circuit is the polynomial associated with the output gate which is defined inductively from the polynomials associated with the nodes feeding into it and the operation at the output gate.
 The size of $\Phi$ is defined to be the number of gates in $\Phi$. For a vertex $v \in V$, we denote the set of variables that occur in the subgraph rooted at $v$ by $X_v$. 

We consider depth restricted circuits. A $\Sigma \Pi \Sigma$ circuit is a levelled depth-$3$ circuit with a plus gate at the top, multiplication gates at the middle level and plus gates at the bottom level. The fan-in of the top plus gate is referred to as top fan-in. A $\Sigma \Pi \Sigma$ circuit is said to be {\em homogeneous} if the plus gate at the bottom level compute homogeneous linear forms only.
 
An important restricted model of arithmetic circuits is multilinear 
circuits. A polynomial $f \in \F[X]$ is called {\em multilinear} if 
the degree of every variable in $f$ is at most one. 
An arithmetic 
circuit is called {\em multilinear} if the polynomial computed 
at every gate is multilinear. 
An arithmetic circuit is called {\em syntactic multilinear} if for every product gate $v$ with children $v_1$ and $v_2$, $X_{v_1} \cap X_{v_2} = \phi$. 

An arithmetic circuit is called an {\em arithmetic formula} 
 if the underlying undirected graph is acyclic i.e. fan-out of every vertex 
 is one. 
 An arithmetic circuit is called {\em skew} if for every product gate, 
 at least one of its children is an input gate. A circuit is called 
 {\em weakly skew} if for every product gate, at least one of its 
 incoming edges is a cut-edge. 

 Let $\Phi$ be a formula defined over the set of variables $X$ and a field 
 $\F$. For a product gate $v$ in $\Phi$ with children $v_1$ and $v_2$, let 
 us define the following properties:
 \begin{description}
 \item[Disjoint] $v$ is said to be {\em disjoint} if $X_{v_1} 
 \cap X_{v_2} = \phi$.
 \item[Sparse] $v$ is said to be {\em $s$-sparse} if the number of 
 monomials in the polynomial computed by at least one of its input gates 
 is at most $2^s$.
 \end{description}
 
 Also, for a node $v$ in $\Phi$, let us define the product-sparse depth of $v$ 
 to be equal to the maximum number of non-disjoint product gates in any path from 
 a leaf to $v$.
 \begin{definition}
 A formula is said to be a $(s, d)$-product-sparse if every product gate $v$ is either disjoint or $s$-sparse,
 where $d$ is the product-sparse depth of the root node.
 \end{definition}
 
 Clearly, any syntactic multilinear formula is a $(s,0)$-product-sparse 
 formula for any $s$. Also, a skew formula is a $(0,d)$-product-sparse 
 formula where $d$ is at most the height of the formula. Thus, proving 
 lower bounds for product-sparse formulas will be a strengthening of 
 known results.
 We also define an extension of the above class of formulas. 
 \begin{definition} 
 A preprocessed product-sparse formula is a product-sparse
 formula in which each input gate which is labelled by an input 
 variable (say, $x_i$) is replaced by a gate labelled by a non-constant 
 univariate polynomial $T(x_i)$ in the same variable. The size of a 
 preprocessed product-sparse formula is defined to be the size 
 of underlying product-sparse formula.
 \end{definition}

 

An arithmetic branching program (ABP) $B$ over a field $\F$ and a set of variables $X$ is defined as a $4$-tuple $(G, w, s, t)$ where $G = (V,E)$ is a directed acyclic graph in which $V$ can be partitioned into levels $L_0, L_1, \dots, L_d$ such that $L_0 = \{s\}$ and $L_d = \{t\}$. 
The edges in $E$ can only go between two consecutive levels. The weight function $w : E \rightarrow X \cup \F$ assigns variables or constants from the field to the edges of $G$. For a path $p$ in $G$, we extend the weight function by $w(p) = \prod_{e \in p}w(e)$. For any $i, j \in V$, let us denote by $P_{i,j}$ the collection of all paths from $i$ to $j$ in $G$. Every vertex $v$ in $B$ computes a polynomial which is given by $\sum_{p \in P_{s,v}} w(p)$. The polynomial $f$ computed by $B$ is defined to be the polynomial computed at the sink $t$ i.e. $f = \sum_{p \in P_{s,t}} w(p)$. The size of $B$ is defined to be $|V|$. The depth of $B$ is defined to be $d$.
  
 For any $i,j \in V$, let us denote by $X_{i,j}$ the set of variables that 
 occur in paths $P_{i,j}$ and denote by $f_{i,j}$ the polynomial 
 $\sum_{p \in P_{i,j}} w(p)$. A homogeneous arithmetic branching program is an ABP $B$ in which the weight function $w$ assigns linear homogeneous forms to the edges of $B$. Clearly, the degree of the homogeneous polynomial computed by $B$ is equal to the depth of $B$.

%

 \begin{definition}
 Let $B = (G,w,s,t)$ be a homogeneous ABP over a field $\F$ and set of variables $X = \{
 x_1, x_2, \dots, x_{2n} \}$. $B$ is said to be $\pi$-partitioned for a 
 permutation $\pi : [2n] \rightarrow [2n]$ if 
 there exists an $i = 2 \alpha n$ for some constant $\alpha$ such 
 that the following condition is satisfied, $\forall v \in L_i: $
 \begin{itemize}
 \item Either, $X_{s,v} \subseteq \{x_{\pi(1)}, x_{\pi(2)}, \dots, x_{\pi(n)} \}$
 and $|X_{v,t}| \le 2n(1 - \alpha)$.
 \item Or, $X_{v,t} \subseteq \{x_{\pi(n+1)}, 
 x_{\pi(n+2)}, \dots, x_{\pi(2n)} \}$ and $|X_{s,v}| \le 2n(1- \alpha)$

 \end{itemize}
 We say that $B$ is partitioned with respect to the level $L_i$.
 $B$ is said to be a partitioned ABP if it is $\pi$-partitioned for some 
 $\pi : [2n] \rightarrow [2n]$.
 \end{definition}

%

%% file: tool.tex
 \section{The Polynomial Coefficient Matrix \& Properties}
 \label{sec:tool}
 In this section, we introduce the main tool used in the paper and prove its properties. Let $Y = \{y_1, y_2, \dots, 
 y_m\}$ and $Z = \{z_1, z_2, \dots, z_m\}$ be two sets of variables. Let $f 
 \in \F[Y, Z]$ be a multilinear polynomial over the field $\F$ and the 
 variables $Y \cup Z$. Define $L_f$ to be the $2^m \times 2^m$ 
 {\em partial derivatives matrix} as follows: for monic 
 multilinear monomials $p \in \F[Y], q \in \F[Z]$, define $L_f(p, q)$ to be the coefficient of the monomial $pq$ in $f$. 
 Let us denote the rank of $L_f$ by $\rank(L_f)$.
 We extend the partial derivatives matrix to non-multilinear polynomials.
 \begin{definition}[Polynomial Coefficient Matrix]
 Let $f \in \F[Y, Z]$ be a polynomial over the field $\F$ 
 and the variables $Y \cup Z$. Define $M_f$ to be the $2^m \times 2^m$ 
 {\em polynomial coefficient matrix} with each entry from the ring $\F[Y,Z]$ defined as follows. For monic multilinear monomials $p$ and $q$ in the set of variables $Y$ and $Z$ respectively, $M_f(p,q) = G$ if and only if $f$ can be uniquely written as $f = pq(G) + Q$, where $G, Q \in \F[Y,Z]$ such that $G$ does not contain any variable other than those present in $p$ and $q$, $Q$ does not have any monomial $m$ which is divisible by $pq$ and which contains only variables that are present in $p$ and $q$.
 \end{definition}
 
\noindent For example, if $f = y_1z_1 + y_1^2z_1 + y_1z_1z_2 + z_1$ then $M_{f}(y_1, z_1) = 1 + y_1$. Observe that we can write, $$f = \sum\limits_{p,q} M_f(p,q)pq ~.$$
 
 Also observe that for a multilinear polynomial $f \in \F[Y,Z]$, the polynomial 
 coefficient matrix $M_f$ is same as the partial derivatives matrix $L_f$.
 For any substitution function $S:Y \cup Z \rightarrow \F$, let us denote by 
 $M_f|_S$ the matrix obtained by substituting each variable to the field 
 element as given by $S$ at each entry in $M_f$. We define {\em max-rank} of 
 $M_f$ as follows:
 \[\maxrank(M_f) = \max\limits_{S:Y \cup Z \rightarrow \F}
   \left\{\rank(M_f|_S) \right\} \]
   
 The following propositions bound the max-rank of the polynomial coefficient 
 matrix.
  
 \begin{proposition}\label{p:2}
 Let $f \in \F[Y,Z]$ be a polynomial over the field $\F$ and the sets of 
 variables $Y' \subseteq Y$ and $Z' \subseteq Z$. Let $a = \min\{|Y'|, |Z'|
 \}$. Then, $\maxrank(M_f) \le 2^a$.
 \end{proposition} 
 \begin{proof}
 In the polynomial coefficient matrix $M_f$, the number of non-zero rows or 
 non-zero columns will be at most $2^a$. Thus, rank of $M_f$ for any 
 substitution would be at most $2^a$. Hence, $\maxrank(M_f) \le 2^a$.
 \end{proof}

 
 \begin{proposition}\label{p:3}
 Let $f, g \in \F[Y,Z]$ be two polynomials. Then,
 \[ \maxrank(M_{f+g}) \le \maxrank(M_f) + \maxrank(M_g). \]
 \end{proposition}
 \begin{proof}
 It is easy to observe that $M_{f+g} = M_f + M_g$. Let $\maxrank(M_{f+g}) = 
 \rank(M_{f+g}|_S)$ for some substitution $S$. Then,
 \begin{eqnarray*}
 \maxrank(M_{f+g}) &=& \rank(M_{f+g}|_S) \\
 &=& \rank(M_f|_S + M_g|_S) \\
 &\le & \rank(M_f|_S) + \rank(M_g|_S) \\
 &\le & \maxrank(M_f) + \maxrank(M_g).
 \end{eqnarray*}
 \end{proof}


 \begin{proposition}\label{p:4}
 Let $Y_1, Y_2 \subseteq Y$ and $Z_1, Z_2 \subseteq Z$ such that $Y_1 \cap 
 Y_2 = \phi$ and $Z_1 \cap Z_2 = \phi$. Let $f \in \F[Y_1, Z_1]$ and $g 
 \in \F[Y_2, Z_2]$. Then,
 \[ \maxrank(M_{fg}) = \maxrank(M_f) \cdot \maxrank(M_g). \]
 \end{proposition}
 \begin{proof}
 We think of $M_f$ as a $2^{|Y_1|} \times 2^{|Z_1|}$ matrix and $M_g$ as a 
 $2^{|Y_2|} \times 2^{|Z_2|}$ matrix as all the other entries are zero. 
 Similarly, we can think of $M_{fg}$ as a $2^{|Y_1 \cup Y_2|} \times 
 2^{|Z_1 \cup Z_2|}$ matrix. Since $f$ and $g$ are defined over disjoint 
 set of variables, we have $M_{fg} = M_f \otimes M_g$ where $\otimes$ 
 denotes the tensor product of two matrices.
 Let $\maxrank(M_{fg}) = \rank(M_{fg}|_S)$ for some substitution $S$. Then,
 \[ \begin{array}{rclcl}
 \maxrank(M_{fg}) &=& \rank(M_{fg}|_S) \\
 &=& \rank((M_f \otimes M_g)|_S) &\le & \rank(M_f|_S \otimes M_g|_S) \\
 &\le & \rank(M_f|_S) \cdot \rank(M_g|_S) &\le & \maxrank(M_f) \cdot \maxrank(M_g).
 \end{array} \]
 Similarly, $\maxrank(M_{fg}) \ge \maxrank(M_f) \cdot \maxrank(M_g)$.
 \end{proof}
 
 \begin{proposition}\label{p:5}
 Let $f \in \F[Y,Z]$ and $g \in \F[Y]$ or $g \in \F[Z]$. Then,
 $\maxrank(M_{fg}) \le \maxrank(M_f)$.
 \end{proposition}
 \begin{proof}
 Without loss of generality, we assume that $g \in \F[Y]$. The case 
 when $g \in \F[Z]$ will follow similarly. 
 For a subset $S \subseteq Y$, we denote the monomial $\Pi_{y \in S}
 y$ by $y^S$. Let us analyze the case when $g = y^{S}$. Consider a row of $M_{fg}$ indexed by the multilinear monomial $p$ in 
 the set of variables $Y$. If $p$ is not divisible by $y^S$, then all 
 the entries in this row will be zero. Otherwise, for any multilinear 
 monomial $q$ in the variables $Z$, we can write,
 $ M_{y^Sf} = \sum_{S' \subseteq S} y^{S \backslash S'} M_f(p/y^{S'}, 
    q) ~.$
 Thus, rows in $M_{y^Sf}$ are a linear combination of rows in 
 $M_f$. Similarly, we can show that for any monomial $m$ in the 
 variables $Y$, rows in $M_{mf}$ are a linear combination of rows in 
 $M_f$.
 
 Now consider any $g = \sum_{i \in [r]} m_i \in \F[Y]$ where $r$ is the 
 number of monomials in $g$ and each $m_i$ is a distinct monomial. Thus, 
 $M_{fg} = \sum_{i \in [r]} M_{m_if}$. Thus, each row in $M_{fg}$ is a 
 linear combination of rows in $M_f$. Hence, $\maxrank(M_{fg}) \le 
 \maxrank(M_f)$.
 \end{proof}
 
 \begin{corollary} \label{c:1}
 Let $f, g \in \F[Y,Z]$. If $g$ is a linear form, then $\maxrank(M_{fg}) \le 2 \cdot \maxrank(M_f)$.
 \end{corollary}
 \begin{proof}
 Since $g$ is a linear form in the variables $Y \cup Z$, $g$ can 
 be expressed as $g = g_1 + g_2$ where $g_1 \in \F[Y]$ and $g_2 \in 
 \F[Z]$ and the proof follows.
 \end{proof}
 
 \begin{corollary} \label{c:2}
 Let $f, g \in \F[Y,Z]$. If $g$ can be expressed as 
 $\sum\limits_{i \in [r]} g_i h_i$ where $g_i \in \F[Y]$ and $h_i 
 \in \F[Z]$, then $\maxrank(M_{fg}) \le r \cdot \maxrank(M_{f})$.
 \end{corollary} 
 \begin{proof}
 Since $M_{fg} = \sum\limits_{i \in [r]} M_{fg_ih_i}$, using Proposition 
 \ref{p:5} completes the proof.
 \end{proof}
 
 \begin{corollary} \label{c:3}
 Let $f, g \in \F[Y,Z]$. If $g$ has $r$ monomials, then 
 	$\maxrank(M_{fg}) \le r \cdot \maxrank(M_{f})$
 \end{corollary}
 \begin{proof}
 Each monomial $m_i$ of $g$ can be written as $g_ih_i$ such that 
 $g_i$ is a monomial in the variables $Y$ and $h_i$ is a monomial 
 in the variables $Z$. Thus, the proof follows using above corollary.
 \end{proof}
 
 \subsection*{Full Rank Polynomials}
 Let $X = \{x_1, \cdots, x_{2n} \}, Y = \{y_1, \cdots, y_n\}$ and $Z = \{z_1, \cdots, z_n\}$ 
 be sets of variables and $f \in \F[X]$. $f$ is said to be a {\em full rank} polynomial if 
 for any partition $A:X \rightarrow Y \cup Z$, $\rank(L_{f^A}) = 2^n$, where $f^A$ is the polynomial 
 obtained from $f$ after substituting every variable $x$ by $A(x)$.
 We say that $f$ is a {\em full max-rank} polynomial if $\maxrank(M_{f^A}) = 2^n$ for any partition $A$. 
 Observe that any full rank polynomial is also a full max rank polynomial. Further more, many full rank 
 polynomials have been studied in the literature~\cite{J08, Raz06, Raz09}.

%% file: homogen.tex
\section{Lower Bounds against Homogeneous Depth-3 Circuits}
\label{sec:depth3}
 
 We recall the definition of homogeneous $\Sigma \Pi \Sigma$ circuits from section~\ref{sec:prelims}. The polynomial computed by a $\Sigma \Pi \Sigma$ circuit with top fan-in $k$ 
 can be represented as $\sum\limits_{i = 1}^{k} P_i$, 
 where $P_i = \prod\limits_{j = 1}^{deg(P_i)}l_{i,j}$, each $l_{i,j}$ is a 
 linear from and $deg(P_i)$ is the fan-in of the $i^{th}$ 
 multiplication gate at the middle level.
 
 Let $\Phi$ be a homogeneous $\Sigma \Pi \Sigma$ circuit defined over the 
 set of variables $X$ and over a field $\F$ computing a homogeneous polynomial $f$. Let us denote 
 the polynomial coefficient matrix of the polynomial computed at the top plus 
 gate of $\Phi$ by $M_{\Phi}$. For a partition $A : X \rightarrow Y \cup Z$, 
 let us denote by $\Phi^A$ the circuit obtained after replacing every variable 
 $x$ by $A(x)$. We prove the following upper bound on the $\maxrank(M_{\Phi^A})$.
 
 \begin{lemma} \label{l:7}
 Let $\Phi$ be a homogeneous $\Sigma \Pi \Sigma$ circuit as defined above. Let the degree of $f$ be equal to $d$. Then, for any partition $A : X \rightarrow  Y \cup Z$, $\maxrank(M_{\Phi^A}) \le k \cdot 2^d$.
 \end{lemma}
 \begin{proof}
 From the definition it is clear that $f$ can be written as: $f = \sum\limits_{i = 1}^k P_i$ 
 where $P_i = \prod\limits_{j = 1}^{deg(P_i)} l_{i,j}$, each $l_{i,j}$ is a homogeneous linear form. 
 Let us denote by $l_{i,j}^A$ and $P_i^A$ the polynomials obtained after substitution 
 of $x$ by $A(x)$ in the polynomials $l_{i,j}$ and $P_i$ respectively.
 
 Since each $l_{i,j}$ is a homogeneous linear form, a multiplication gate $P_i$ 
 computes a homogeneous polynomial of degree $deg(P_i)$. Thus if $deg(P_i) \neq d$ then 
 the multiplication gate $P_i$ does not contribute any monomial in the output 
 polynomial $f$. Hence, it can be assumed without loss of generality that $deg(P_i) = d$ 
 for all $i \in [k]$.

 Since $l_{i,j}$ is a homogeneous linear form, $\maxrank(M_{l_{i,j}^A}) \le 2$. Thus, 
 using Corollary \ref{c:1},$\forall i \in [k] : \maxrank(M_{P_i^A}) \le 2^d$.
 Hence, using Proposition \ref{p:3},
 $\maxrank(M_{f^A}) \le \sum_{i \in [k]} \maxrank(M_{P_i^A}) \le k \cdot 2^d$.
 \end{proof}
 
 In \cite{NW95}, it was proved that any homogeneous $\Sigma \Pi \Sigma$ circuit for multiplying 
 $d$ $n \times n$ matrices requires $\Omega(n^{d-1}/d!)$ size. We prove a better lower 
 bound using our techniques. To consider a single 
 output polynomial, we will concentrate on the $(1,1)^{th}$ entry of the product.
%
 Formally, let $X^1, X^2, \dots, X^d$ be disjoint sets of variables of size $n^2$ each, with $X = 
 \cup_{i \in [d]}X^i$. The variables in $X^i$ will be denoted by $x^i_{jk}$ for $j, k 
 \in [n]$. We will be looking at the problem of multiplying $d$ $n \times n$ matrices 
 $A^1, A^2, \dots, A^d$ where $(j,k)^{th}$ entry of matrix $A^i$, denoted by $A^i_{jk}$, 
 is defined to be equal $x^i_{jk}$ for all $i \in [d]$ and $j, k \in 
 [n]$. The output polynomial, that we are interested in, is the $(1,1)^{th}$ entry of 
 $\prod_{i \in [d]}A^i$ denoted by $f$. $f$ is clearly a homogeneous multilinear 
 polynomial of degree $d$. Moreover, any monomial in $f$ contains one variable each from the 
 sets $X^1, X^2, \dots, X^d$.
 
  We first prove an important lemma below. We also provide an alternative induction based 
  proof for the below lemma in the Appendix \ref{app:l-8}.
 
 \begin{lemma} \label{l:8}
 For the polynomial $f$ as defined above, there exists a bijective partition 
 $B : X \rightarrow Y \cup Z$ such that $\maxrank(M_{f^B}) = n^{d-1}.$
 \end{lemma}
 \begin{proof}
 We fix some notations first. For $i < j$, let us denote the set $\{ i, i + 1, \ldots, j \}$ by $[i, j]$. Let us also
 denote the pair $\left( (k, i), (k, j) \right)$ by $e_{ijk}$ for any $i, j, k$.
 Construct a directed graph $G(V, E)$ on the set of vertices $V = [0,d] \times [1,n] $
 and consisting of edges $E = \left\{  e_{ijk} \mid k \in [0, d - 1], i, j \in [1,n] \right\}$.
 Note that the edges $e_{ijk}$ and $e_{jik}$ are two distinct edges for fixed values of $i,j,k$
 when $i \neq j$.
 Let us also define a weight function $w: E \rightarrow X$ such that $w( e_{ijk} ) = x^{k+1}_{ij}$.
 
 It is easy to observe that the above graph encodes the matrices $A_1, A_2, \ldots, A_d$. The weights 
 on the edges are the variables in the matrices. For example, a variable $x^{k+1}_{ij}$ in the matrix 
 $A_{k+1}$ is the weight of the edge $e_{ijk}$. Let us denote the set of paths in $G$ from the 
 vertex $(0,1)$ to the vertex $(d,1)$ by ${\cal P}$. Let us extend the weight function and
 define $w( p ) = \prod_{e \in p}{ w(e) }$ for any $p \in {\cal P}$. Since, all paths in ${\cal P}$
 are of length equal to $d$, the weights corresponding to each of these paths are monomials of 
 degree $d$.
 
 Let us define the partition $B : X \rightarrow Y \cup Z$ as follows: all the variables in odd 
 numbered matrices are assigned variables in $Y$ and all the variables in even numbered 
 matrices are assigned variables in $Z$. Let us denote the variable assigned by $B$ to 
 $x^{2k - 1}_{ij}$  by $y^{2k - 1}_{ij}$ and the variable assigned to $x^{2k}_{ij}$ by 
 $z^{2k}_{ij}$.

 It follows from the matrix multiplication properties that for any path $p \in {\cal P}$, the
 monomial $w(p)$ is a monomial in the output polynomial. Each such path is uniquely specified 
 once we specify the odd steps in the path. Now, specifying odd steps in the path corresponds 
 to specifying a variable from each of the odd numbered matrices. To count number of such ways, 
 let us first consider the case when $d$ is even. There are $d/2$ odd numbered matrices and we 
 have $n^2$ ways to choose a variable from each of these $d/2$ matrices except for the first matrix 
 for which we can only choose a variable from the $1^{st}$ row since our output polynomial is 
 the $(1,1)^{th}$ entry. Thus, there are $n^{d-1}$ 
 number of ways to specify one variable each from the odd numbered matrices, the number of 
 such paths is also $n^{d-1}$. We get the same count for the case when $d$ is odd using the 
 similar argument. Since once the odd steps are chosen, there is only one way to 
 choose the even steps, all these $n^{d-1}$ monomials give rise to non-zero entries in different
 rows and columns in the matrix $M_{f^B}$. Hence, the matrix is an identity block of dimension $n^{d-1}$ upto a permutation
 of rows and columns and thus it has rank equal to $n^{d-1}$.
 \end{proof}
 
 \begin{theorem} \label{t:3}
 Any homogeneous $\Sigma \Pi \Sigma$ circuit for computing the product of $d$ $n \times n$ matrices requires 
 $\Omega(n^{d-1}/2^d)$ size.
 \end{theorem}
  
 \begin{proof}
 Let $\Phi$ be a homogeneous depth-$3$ circuit computing $f$. Then, using Lemma~\ref{l:7}, for any partition $A$, $\maxrank(M_{f^A}) \le k \cdot 2^d.$
 From Lemma \ref{l:8}, we know that there exists a partition $B$ such that
 $\maxrank(M_{f^B}) = n^{d-1}$. Hence, $k \ge n^{d-1}/2^d$.
 \end{proof}

 It is worth noting that there exists a depth-$2$ circuit of size $n^{d-1}$ computing IMM polynomial.
 As observed in Lemma~\ref{l:8}, there are $n^{d-1}$ monomials in the IMM polynomial. 
 Hence, the sum of monomials representation for IMM will have top fan-in equal to $n^{d-1}$. 
 We remark that when the number of matrices is a constant, the upper and lower bounds for IMM polynomial match.

%

%% file: lowrank.tex
\section{Lower Bounds against Depth-3 Circuits of Bounded Product Dimension}
\label{sec:lowrank}

If a depth-$3$ circuit is not homogeneous, the fan-in of a product gate can be arbitrarily larger than the degree of the polynomial being computed. Hence the techniques in the previous section fails to give non-trivial circuit size lower bounds.
In this section, we study depth-$3$ circuits with bounded product dimension - where the affine forms feeding into every product gate are from a linear vector space of small dimension and prove exponential size lower bounds for such circuits. 





We will first prove an upper bound on the maxrank of the polynomial coefficient matrix for the polynomial computed by a depth-$3$ circuit of product dimension $r$, parameterized by $r$. Let $C$ be a $\Sigma\Pi\Sigma$
circuit of product dimension $r$ and top fan in $k$. 
Let $P^j$ be the product gates in $C$ for $j \in [k]$, given by $P^j = \Pi_{i = 1}^s{L_i^j}$. Without loss of generality, let us assume that the 
vectors $L_1^j, L_2^j,\ldots, L_r^j$ form a basis for the span of $\{L_1^j, L_2^j,\ldots,L_s^j\}$. Let $l_i^j$ be the homogeneous part of $L_i^j$ for each $i$. So, clearly 
the set $\{l_i^j\}_{i \in [r^{'}]}$ spans the set $\{l_i^j\}_{i \in [s]}$, where $r^{'} \leq r$. To simplify the notation, we will refer to $r^{'}$ as $r$.
In the following presentation, we will always use $d$ to refer to the
degree of the homogeneous polynomial computed by the circuit under consideration. Now, let us express each $l_i^j$ as a linear 
combination of $\{l_i^j\}_{i\in [r]}$. Let us now expand the product $P^j$ into a sum of product of homogeneous linear forms coming from $\{l_i^j\}_{i\in r}$. 
Let $P_d^j$ be the slice of $P^j$ of degree exactly $d$, for each $j \in [k]$. We now have the following observation. 

\begin{observation}
Let $C_d = \Sigma_{i \in [k]}{P_d^i}$. If $C$ computes a homogeneous polynomial of degree $d$, 
then $C_d$ computes the same polynomial. 
\end{observation}
\begin{proof}
 The proof follows from the fact that since $C$ computes a homogeneous polynomial of degree $d$, 
the monomials for degree other than $d$ cancel each other across the different product gates.
\end{proof}

We now look at each product in $P_d^j$, which is a sum of products. Each such product is a product of homogeneous linear forms from $\{l_i^j\}_{i\in [r]}$ of degree exactly $d$. To simplify it further, we will use the following lemma.

\begin{lemma}(\cite{Shp01})\label{lem:sopdecomposition}
Any monomial of degree $d$ can be written as a sum of $d^{th}$ power of some $2^d$ linear forms. Further, each of the $2^d$ linear forms in the expression corresponds to $\Sigma_{x\in S}{x}$ for a subset $S$ of $[d]$.
\end{lemma}

By applying this lemma to each product term in the sum of product representation of $P_d^j$, we obtain the following:

\begin{lemma}{\label{lem:soplin}}
 If $P_d^j = \Sigma_i\Pi_{u=1}^d{l_{\alpha_{iu}}^j}$ where $\alpha_{iu} \in [r]$, then $P_d^j = \Sigma_{q=1}^v{c_qL_q}^d$ for some homogeneous linear forms $L_q$, constants $c_q$ and
$v \leq {{d+r}\choose{r}}$.
\end{lemma}
 
\begin{proof}
Consider any product term in the sum of products expansion $P_d^j$ as described, say $S = \Pi_{u=1}^d{l_{\alpha_{iu}}^j}$. From Lemma \ref{lem:sopdecomposition}, we know that $S$ can be written 
as $S = \Sigma_{t=1}^{2^d} {L_t}^d$, where for every subset $U$ of $[d]$, there is a $\beta \in [2^d]$ such that $L_{\beta} = \Sigma_{u\in U}{l_{\alpha_{iu}}^j}$.  In general, each $L_t$ can be 
written as $L_t = \Sigma_{i\in [r]}{\gamma_i l_i^j}$ for non-negative integers $\gamma_i$ satisfying $\Sigma_{i\in [r]}{\gamma_i \leq d}$. Now, each of the product terms in $P_d^j$ can
be expanded in a similar fashion into $d^{th}$ powers of linear forms, each from the set $\{\Sigma_{i\in [r]}{\gamma_i l_i^j}: \gamma_i\in {\mathbb Z^{\geq 0}}\wedge \Sigma_{i\in [r]}{\gamma_i \leq d}\}$. The number of distinct such linear forms is at most ${{d+r}\choose{r}}$. Hence, the lemma follows. 
\end{proof}

We now bound the maxrank of the power of a homogeneous linear form
which in turn will give us a bound for
$P_d^j$ due to the subadditivity of maxrank. 

\begin{lemma}\label{lem:coplin}
 Given a linear form $l$ and any positive integer $t$, the maxrank of $l^t$ is at most $t + 1$ 
for any partition of the set $X$ of variables into $Y$ and $Z$.
\end{lemma}

\begin{proof}
Partition the linear form $l$ into two parts,  $l = l_y + l_z$, where $l_y$ consists of all variables 
in $l$ from the set $Y$ and $l_z$ consists of the variables which come from the set $Z$. By the binomial theorem, $l^t = \Sigma_{i=0}^t {t\choose i}l_y^i l_z^{t-i}$. Now, $l_y^i$ is a polynomial just in $Y$ variables and hence its maxrank can be bounded above by $1$, and multiplication by $l_z^{t-i}$ does not increase the maxrank any further, by proposition~\ref{p:5}. Hence, the maxrank of each term in the sum is at most $1$ and there are at most $t+1$ terms, so, by using the subadditivity of maxrank, we get an upper bound of $t+1$ on the maxrank of the sum.
\end{proof}

Now we are all set to give an upper bound on the maxrank of $P_d^j$. 

\begin{lemma}
 The max rank of $P_d^j$ is at most $(d+1){{d+r}\choose{r}}$ for any partition of the set $X$ of variables into $Y$ and $Z$.
\end{lemma}

\begin{proof}
The proof follows from Lemma \ref{lem:soplin}, Lemma \ref{lem:coplin} and the subadditivity of max rank.  
\end{proof}

Now we are ready to prove the theorem.
\begin{theorem}
\label{thm:lowrank}
There is an explicit polynomial in $n$ variables $X$ and degree at most $\frac{n}{2}$ for which any $\Sigma\Pi\Sigma$ circuit $C$ of product dimension at most $\frac{n}{10}$ requires size $2^{\Omega(n)}$.
\end{theorem}
\begin{proof}
We describe the explicit polynomial $Q(X)$ first. Fix an equal sized partition $A$ of $X$ into $Y$ and $Z$.
Order all subsets of $Y$ and $Z$ of size exactly $\frac{n}{4}$ in any order, say $S_1, S_2, \ldots, S_w$ and $T_1, T_2, \ldots, T_w$, where $w = {{\frac{n}{2}}\choose{\frac{n}{4}}}$. Let us define the polynomial $Q^A(Y,Z)$ for the partition $A$ as follows: $$Q^A(Y,Z) = {\Sigma_{i = 1}^w{{\Pi_{y\in S_i}}{\Pi_{z\in T_i}}yz}}$$ We obtain the polynomial $Q(X)$ by replacing variables in $Y$ and $Z$ in $Q^A(Y,Z)$ by $A^{-1}(Y)$ and $A^{-1}(Z)$ respectively. The polynomial $Q(X)$ is homogeneous and of degree $\frac{n}{2}$.

Now we prove the size lower bound. The polynomial coefficient matrix of $Q$ with respect to the partition $Y$ and $Z$ is simply the diagonal submatrix, and the rank is at-least $\frac{2^{\frac{n}{2}}}{\sqrt{n}}$ 
Since the circuit computes the polynomial, the top fan in $k$ should be at least $\frac{\frac{2^{\frac{n}{2}}}{\sqrt{n}}}{{{d+r}\choose{r}}{(d+1)}}$, for $d = \frac{n}{2}$, and product dimension $\frac{n}{10}$, we have a lower bound of $2^{cn}$, for a constant $c > 0$.
\end{proof}

\paragraph{An Impossibility result:}
Consider the trivial depth-$2$ circuit for any polynomial, where each monomial is computed by the product gate. Viewing this as a depth-$3$ circuit, the total dimension of the circuit is bounded above by $n$, since there are only $n$ variables. Can we have a circuit with a smaller total dimension $r$ computing the same polynomial? We show that this is not always possible if  $r = \alpha.n$ for a sufficiently small constant $\alpha < 1$. In particular, we show that even for $r = \frac{n}{10}$, they cannot compute the polynomial that we constructed in the proof of theorem~\ref{thm:lowrank} irrespective of the size of the circuit. 
As a first step, using ideas developed in the previous subsection, we prove the following upper bound for maxrank of such circuits. 

\begin{lemma}
If the total dimension of a $\Sigma\Pi\Sigma$ circuit is $r$, then the maxrank
of the polynomial computed by the circuit is bounded above by  ${{d+r}\choose{r}}{(d+1)}$.
\end{lemma}
\begin{proof}
 Observe that if the span of all the affine forms occurring
in a depth-$3$ $\Sigma\Pi\Sigma$ circuit is $r$ (spanned by the basis $L_1, L_2,\ldots, L_r$), then each of the product gates in the circuit can be 
decomposed into sum of power of homogeneous linear forms as in Lemma \ref{lem:soplin}. Moreover, each of these homogeneous linear forms will be of the form 
$\Sigma_i \alpha_i l_i$, where $\alpha_i\in {\mathbb Z^{\geq0}} \wedge \Sigma_i{\alpha_i} \leq d$ and $l_i$ is the homogeneous part of $L_i$ for each $i$ in $[r]$. 
Consequently, the maxrank for the circuit is bounded by ${{d+r}\choose{r}}{(d+1)}$ by Lemma \ref{lem:coplin} and the subadditivity of max rank. 
\end{proof}

Thus, a $\Sigma\Pi\Sigma$ circuit of total dimension bounded by $r$, can compute the polynomial $Q$ described in the proof of \ref{thm:lowrank}, only if $${{d+r}\choose{r}}{(d+1)} \geq \frac{2^{\frac{n}{2}}}{\sqrt{n}}~.$$ This in turn implies that for $r \leq \frac{n}{10}$, such circuits cannot compute the polynomial $Q$ irrespective of the number of gates they use.

%% file: prod-sparse.tex
\section{Lower Bounds against Product-sparse Formulas} 
\label{sec:product-sparse}
 Let $Y = \{y_1, y_2, \dots, y_m\}$ and $Z = \{z_1, z_2, \dots, z_m\}$.
 Let $\Phi$ be an arithmetic circuit defined over the field 
 $\F$ and the variables $Y \cup Z$. For a node $v$, let 
 us denote by $\Phi_v$ the sub-circuit rooted at $v$, and  denote by 
 $Y_v$ and $Z_v$, the set of 
 variables in $Y$ and $Z$ that appear in $\Phi_v$ respectively.
 Let us define, $a(v) = \min\{|Y_v|, |Z_v|\}$ and $b(v) = (|Y_v| + 
 |Z_v|)/2$. We say that a node $v$ is $k$-unbalanced if $b(v) - a(v) 
 \ge k$.
 Let $\gamma$ be a simple path from a leaf to the node $v$. We say that 
 $\gamma$ is $k$-unbalanced if it contains at least one $k$-unbalanced 
 node. We say that $\gamma$ is central if for every $u, u_1$ on the path 
 $\gamma$ such that there is an edge from $u_1$ to $u$ in $\Phi$, $b(u) 
 \le 2b(u_1)$. $v$ is said to be $k$-weak 
 if every central path that reaches $v$ is $k$-unbalanced.

 We prove that if $v$ is $k$-weak then the $\maxrank$ of the 
 matrix $M_v$ can be bounded. The proof goes via induction on $|\Phi_v|$ and 
 follows the same outline as that of \cite{Raz09}. It only differs in the case 
 of non-disjoint product gates which we include in full detail below. The proofs 
 of the rest of cases is given in the appendix~\ref{app:repeat-proof}.
 \begin{lemma} \label{l:2}
 Let $\Phi$ be a $(s,d)$-product-sparse formula over the 
 set of variables $\{y_1, \dots, y_m\}$ and $\{z_1, \dots, z_m\}$, and 
 let $v$ be a node in $\Phi$. Denote the product-sparse depth of $v$ by $d(v)$. If $v$ is $k$-weak, then,
 $\maxrank(M_v) \le 2^{s \cdot d(v)} \cdot |\Phi_v| \cdot 2^{b(v)-k/2} ~.$ 
 \end{lemma}
 \begin{proof}
  Consider the case when $v$ is a $s$-sparse product gate with children $v_1$ and $v_2$ and $v$. 
  Without loss of generality it can be assumed that $v$ is not disjoint.
  
  Let us suppose that the product-sparse depth 
  of $v$ is $d$. Without loss of generality, assume that $v_2$ computes 
  a sparse polynomial having at most $2^s$ number of monomials. Since 
  $v$ is not disjoint, product-sparse depth of $v_1$ is at most 
  $d - 1$. Thus using Corollary \ref{c:3},
  \begin{eqnarray} \label{eqn:1}
  \maxrank(M_v) \le 2^s \cdot \maxrank(M_{v_1}) 
  \end{eqnarray}
  Consider the following cases based on whether $b(v) \le 2b(v_1)$ or 
  not.
  
  If $b(v) \le 2b(v_1)$, then $v_1$ is also $k$-weak. Therefore, by 
  induction hypothesis, 
  \[ \maxrank(M_{v_1}) \le 2^{s(d-1)} \cdot |\Phi_{v_1}| \cdot 
     2^{b(v_1) - k/2} \le 2^{s(d-1)} \cdot |\Phi_{v}| \cdot 
     2^{b(v) - k/2} ~.\]
  Thus, using Equation \ref{eqn:1},
  $ \maxrank(M_v) \le 2^{sd} \cdot |\Phi_v| \cdot 2^{b(v) - k/2} ~. $
  If $b(v) > 2b(v_1)$, then $b(v_1) < b(v)/2 < b(v) - k/2$ since $b(v) 
  \ge k$. Therefore using Proposition \ref{p:2}, $\maxrank(M_{v_1}) 
  \le 2^{a(v_1)} \le 2^{b(v_1)} < 2^{b(v) - k/2}$. Therefore, 
  $\maxrank(M_v) \le 2^s \cdot 2^{b(v) - k/2} \le 2^{sd} \cdot |\Phi_v| 
  \cdot 2^{b(v) - k/2} ~.$
%
%
 \end{proof}

 Because of previous lemma, to prove that a full max-rank polynomial 
 cannot be computed by any $(s,d)$-product-sparse formula of 
 polynomial size, we only need to show that there exists a partition 
 that makes the formula $k$-weak with suitable values of $s, d$ and $k$.

 In \cite{Raz06}, it was proved that for syntactic multilinear formulas of 
 size at most $n^{\epsilon \log n}$, where $\epsilon$ is a small 
 enough universal constant, there exists such a partition that makes 
 the formula $k$-weak for $k = n^{1/8}$. We observe that this lemma 
 also holds for product-sparse formulas, the proof given in \cite{Raz06} 
 is not specific to just syntactic multilinear formulas and holds for any 
 arithmetic formula.
 We state the lemma again for the case of product-sparse formulas.
 
 \begin{lemma} \label{l:3}
 Let $n = 2m$. Let $\Phi$ be a $(s,d)$-product-sparse formula over the 
 set of variables $X = \{x_1, \ldots, x_n \}$, such that every variable 
 in $X$ appears in $\Phi$, and such that $|\Phi| \le n^{\epsilon 
 \log n}$, where $\epsilon$ is a small enough universal constant. Let 
 $A$ be a random partition of the variables in $X$ into $\{y_1, \ldots, 
 y_m\} \cup \{z_1, \ldots, z_m\}$. Then with probability 
 of at least $1 - n^{-\Omega(\log n)}$ the formula $\Phi^A$ is $k$-weak, 
 for $k = n^{1/8}$.
 \end{lemma} 
 
 With above lemma, the following theorem becomes obvious.
 
 \begin{theorem} \label{t:1}
 Let $X$ be a set of $2n$ variables and let $f \in \F[X]$ be a full 
 max-rank polynomial. Let $\Phi$ be any $(s,d)$-product-sparse formula 
 of size $n^{\epsilon \log n}$, where $\epsilon$ is the 
 same constant for which Lemma \ref{l:3} holds. If $sd = o(n^{1/8})$, 
 then $f$ cannot be computed by $\Phi$.
 \end{theorem}
 \begin{proof}
 Assume for a contradiction that $\Phi$ computes $f$. 
 Using Lemma \ref{l:3}, for a random partition $A$, with 
 probability of at least $1 - n^{-\Omega(\log n)}$, the formula $\Phi^A$ 
 is $k$-weak for $k = n^{1/8}$. Hence, using Lemma \ref{l:2},
 $\maxrank(M_{\Phi^A}) \le 2^{sd} \cdot |\Phi^A| \cdot 2^{n - k/2} < 2^n.$ Since $f$ is a full max-rank polynomial, it cannot be computed by $\Phi$.
 \end{proof}
 
\subsection*{Preprocessed Product-sparse Formulas}
To prove the results about preprocessed product-sparse formulas, we observe first that the Lemma \ref{l:2} also holds for preprocessed product-sparse formulas.

 \begin{lemma} \label{l:4}
 Let $\Phi$ be a preprocessed $(s,d)$-product-sparse formula over the 
 set of variables $\{y_1, \dots, y_n\}$ and $\{z_1, \dots, z_n\}$, and 
 let $v$ be a node in $\Phi$. If $v$ is $k$-weak, then,
 \[ \maxrank(M_v) \le 2^{s \cdot d(v)} \cdot |\Phi_v| \cdot 2^{b(v)-k/2}. \]
 where $d(v)$ is the product-sparse depth of $v$.
 \end{lemma}
 \begin{proof}
 The proof proceed in the same way by induction on $|\Phi_v|$. We only 
 point out the differences in the two proofs. Base case will hold 
 similarly as for any univariate polynomial $T(x)$, $\maxrank(M_{T(x)})$ 
 is at most one. In the induction step, only difference will be in Case $(3)$ 
 in which $v$ is a product gate obeying $s$-sparse property. In lemma 
 \ref{l:2}, we obtain the following inequality using Corollary \ref{c:3},
 \[ \maxrank(M_v) \le 2^s \cdot \maxrank(M_{v_1}) ~. \]
 We obtain the same inequality for preprocessed product-sparse formulas also 
 using Corollary \ref{c:2} instead of Corollary \ref{c:3} and rest of the 
 proof follows similarly.
 \end{proof}
 
 We also observe that if a product-sparse formula $\Phi$ is $k$-weak, 
 then any preprocessed product-sparse formula obtained from 
 $\Phi$ by applying a preprocessing step is also $k$-weak.
 
 \begin{lemma} \label{l:5}
 Let $\Phi$ be a product-sparse formula over the set of variables 
 $Y \cup Z$. Let $\Phi'$ be any preprocessed product-sparse 
 formula obtained from $\Phi$. If $\Phi$ is $k$-weak, then $\Phi'$ is 
 also $k$-weak. 
 \end{lemma}
 \begin{proof}
 For every node $v$ in $\Phi$, there is a corresponding node $v'$ in 
 $\Phi'$. To prove the lemma we need to show that every central path 
 in $\Phi'$ is $k$-weak. Since in the preprocessing step, each variable is 
 replaced by a non-constant univariate polynomial in the same variable, 
 we know that $Y_{v'} = Y_v$ and  $Z_{v'} = Z_v$. Thus, $a(v') = a(v)$ 
 and $b(v') = b(v)$. Hence, every central path in $\Phi'$ is a central 
 path in $\Phi$ and vice-versa. Also, $v$ is $k$-unbalanced in $\Phi$ 
 if{f} $v'$ is also $k$-unbalanced in $\Phi'$. Hence, if $\Phi$ is $k$-weak 
 then $\Phi'$ is also $k$-weak.
 \end{proof}
 
 By using Lemma~\ref{l:4}, in a very similar way to the proof of Theorem~\ref{t:1}, we get the following:
 \begin{theorem} \label{t:2}
 Let $X$ be a set of $2n$ variables and let $f \in \F[X]$ be a full 
 max-rank polynomial. Let $\Phi$ be any preprocessed $(s,d)$-product-sparse formula 
 of size $n^{\epsilon \log n}$, where $\epsilon$ is the 
 same constant for which Lemma \ref{l:3} holds. If $sd = o(n^{1/8})$, 
 then $f$ cannot be computed by $\Phi$.
 \end{theorem}
 \begin{proof}
 Assume for a contradiction that $\Phi$ computes $f$. 
 Using Lemma \ref{l:3} and \ref{l:5}, for a random partition $A$, with 
 probability of at least $1 - n^{-\Omega(\log n)}$, the formula $\Phi^A$ 
 is $k$-weak for $k = n^{1/8}$. Hence, using Lemma \ref{l:4},
 \[ \maxrank(M_{\Phi^A}) \le 2^{sd} \cdot |\Phi^A| \cdot 2^{n - k/2} < 2^n. \]
 Since $f$ is a full max-rank polynomial, it cannot be computed by $\Phi$.
 \end{proof} 

%% file: part-abp.tex
 
\section{Lower Bounds against Partitioned Arithmetic Branching Programs}
\label{sec:partabp}

In the preliminaries section, we defined partitioned arithmetic branching programs which are a generalization of ordered ABPs. While ordered ABP can only compute multilinear polynomials, partitioned ABP is a non-multilinear model, thus can compute non-multilinear polynomials also. We, then, prove an exponential lower bound for partitioned ABPs.
 
By definition, any polynomial computed by a partitioned ABP is homogenous. 
In \cite{J08}, a full rank homogenous polynomial was constructed. Thus, 
to prove lower bounds for partitioned ABP, we only need to upper bound 
the $\maxrank$ of the polynomial coefficient matrix for any polynomial 
being computed by a partitioned ABP. Now we prove such an upper bound and use it to prove exponential lower bound on the size of partitioned ABPs computing any full max-rank homogenous polynomial.

\begin{theorem}
 Let $X$ be a set of $2n$ variables and $\F$ be a field. For any full max-rank 
 homogenous polynomial $f$ of degree $n$ over $X$ and $\F$, the size of any 
 partitioned ABP computing $f$ must be $2^{\Omega(n)}$.
\end{theorem}
 \begin{proof}
 Let $B$ be a $\pi$-partitioned ABP computing $f$ for a permutation $\pi : [2n] 
 \rightarrow [2n]$. Let $L_0, L_1, \dots, L_n$ be the levels of $B$. Consider 
 any partition $A$ that assigns all $n$ $y$-variables to $\{x_{\pi(1)}, 
 x_{\pi(2)}, \dots, x_{\pi(n)} \}$ and all $n$ $z$-variables to $\{x_{\pi(n+1)}, 
 x_{\pi(n+2)}, \dots, x_{\pi(2n)} \}$. Let us denote by $f^A$ the polynomial 
 obtained from $f$ after substituting each variable $x$ by $A(x)$. Let $B$ is 
 partitioned with respect to the level $L_i$ for $i = 2\alpha n$. We can write, 
 $ f = f_{st} = \sum_{v \in L_i} f_{s,v} f_{v,t} ~. $
 Consider a node $v \in L_i$. By definition, there are following two cases:\\
 {\bf Case 1:} $X_{s,v} \subseteq \{x_{\pi(1)}, x_{\pi(2)}, \dots, x_{\pi(n)} \}$
 and $|X_{v,t}| \le 2n(1 - \alpha)$. Thus, $f_{s,v}^A \in \F[Y]$. Hence, using 
 Proposition \ref{p:5},
\[ \maxrank(M_{f_{s,v}^A f_{v,t}^A})  ~\le ~\maxrank(M_{f_{v,t}^A}) 
 ~\le ~2^{|X_{v,t}|/2} ~ \le ~ 2^{n(1 - \alpha)} \]
 {\bf Case 2:} $X_{v,t} \subseteq \{x_{\pi(n+1)}, x_{\pi(n+2)}, \dots, x_{\pi(2n)} 
 \}$ and $|X_{s,v}| \le 2n(1- \alpha)$. Thus, $f_{v,t}^A \in \F[Z]$. Hence, again 
 using Proposition \ref{p:5},
\[ \maxrank(M_{f_{s,v}^A f_{v,t}^A}) ~\le ~\maxrank(M_{f_{s,v}^A}) 
 ~\le ~2^{|X_{s,v}|/2}  ~ \le ~ 2^{n(1 - \alpha)} \]
 Thus, in any case, $\maxrank(M_{f_{s,v}^A f_{v,t}^A}) \le 2^{n(1 - \alpha)}$ for 
 all $v \in L_i$. Using Proposition \ref{p:3},
 $ \maxrank(M_{f^A}) \le |L_i| \cdot 2^{n(1 - \alpha)}. $
 Since $f$ is a full max-rank polynomial, we get $|L_i| \ge 2^{\alpha n}$. 
 \end{proof}

%% file: appendix.tex
 \section{An alternative Proof of Lemma~\ref{l:8} }
 \label{app:l-8}
 We view the product of $d$ matrices as an iterative process using the associative property of the matrix multiplication and prove a more stronger
 statement than Lemma \ref{l:8}. 
 We will first prove it for the case when $d$ is an even integer.
 Let us define a partition $B : X \rightarrow Y \cup Z$ as follows: all the variables in odd 
 numbered matrices are assigned variables in $Y$ and all the variables in even numbered 
 matrices are assigned variables in $Z$. Let us denote the variable assigned by $B$ to 
 $x^{2i - 1}_{jk}$  by $y^{2i - 1}_{jk}$ and the variable assigned to $x^{2i}_{jk}$ by 
 $z^{2i}_{jk}$.
 
 \begin{lemma} \label{l:11}
 Let $d = 2d'$ be an even integer. Let us denote the polynomial computed at $(i,j)^{th}$ entry of 
 the product $\prod\limits_{k \in [d]} A^k$ by $f_{ij}$. The following statements hold true:
 \begin{enumerate}
 \item For any $i, j \in [n]$, $\rank(M_{f_{ij}^B}) = n^{d-1}$.
 \item For any $i \in [n]$ and $j \neq j' \in [n]$, the set of non-zero columns in 
       $M_{f_{ij}^B}$ and $M_{f_{ij'}^B}$ are disjoint.
 \end{enumerate}
 \end{lemma}
 \begin{proof}
 We will prove by induction on $d'$. \\
 {\bf Base Case:} $d' = 1$ i.e. $d = 2$. In this case, we have two matrices $A^1$ and $A^2$ and 
 the partition $B$ assigns all the variables in $A^1$ to $Y$ and all the variables in $A^2$ to $Z$.
 For any $i, j$, $f_{ij} = (A^1A^2)(i, j) = \sum\limits_{k \in [n]} A^1(i,k) A^2(k,j)$. Thus, $f_{ij}^B = 
 \sum\limits_{k \in [n]} y^1_{ik} z^2_{kj}$. Clearly, $\rank(M_{f_{ij}^B}) = n$. Moreover, there 
 are only $n$ non-zero columns in $M_{f_{ij}^B}$ which are indexed by $z^2_{kj}$ respectively.
 Thus, for $j \neq j'$, the set of non-zero columns in $M_{f_{ij}^B}$ and $M_{f_{ij'}^B}$ are disjoint. \\
 {\bf Induction Step:} Let us suppose the lemma holds for $(d' - 1)$. We will prove that it 
 also holds for $d'$. Let $d = 2d'$, $Q = \prod\limits_{k \in [d - 2]} A^k$ and $P = A^{d-1}A^{d}$.
 Therefore, $\prod\limits_{k \in [d]} A^k = QP$. Thus, $f_{ij} = (QP)(i, j) = 
 \sum\limits_{k \in [n]} Q(i,k)P(k,j)$ where $P(k,j)$ can again be written as 
 $\sum\limits_{r \in [n]} A^{d-1}(k,r) A^{d}(r,j)$. Thus, 
 \[ f_{ij} = \sum\limits_{k \in [n]} \sum\limits_{r \in [n]} Q(i,k) ~ A^{d-1}(k,r) ~ A^{d}(r,j). \] 
 Let us denote the polynomial $Q(i, k)$ by $g_{ik}$ and 
 $\sum\limits_{r \in [n]} g_{ik}^B ~ y^{d-1}_{kr} ~ z^{d}_{rj}$ by $P_k$. Thus, 
 \[ f_{ij}^B = \sum\limits_{k \in [n]} \sum\limits_{r \in [n]} g_{ik}^B ~ y^{d-1}_{kr} ~ z^{d}_{rj} 
    = \sum\limits_{k \in [n]} P_k ~. \]
 
 By induction hypothesis, $\rank(M_{g_{ik}^B}) = n^{d-3}$. Thus, $M_{g_{ik}^B}$ has a sub-matrix of 
 size $n^{d-3} \times n^{d-3}$ which is of full rank. Since the variables $y^{d-1}_{kr}$ and $z^{d}_{rj}$ 
 do not appear in $g_{ik}^B$, the matrix $M_{g_{ik}^B ~ y^{d-1}_{kr} ~ z^{d}_{rj}}$ contains the full 
 rank submatrix of $M_{g_{ik}^B}$ at a shifted position in both rows and columns depending on 
 $y^{d-1}_{kr}$ and $z^{d}_{rj}$. 
 Formally, for any monomials $p, q$ in the set of variables $Y$ and $Z$ respectively,  
 $M_{g_{ik}^B ~ y^{d-1}_{kr} ~ z^{d}_{rj}}(y^{d-1}_{kr}p, ~ z^{d}_{rj}q) = 1$ if{f}
 $M_{g_{ik}^B}(p, q) = 1$.
 Thus, $M_{P_k}$ contains $n$ disjoint copies 
 of the full rank sub-matrix of $M_{g_{ik}^B}$ such that no two of the copies contain any common 
 non-zero row or column. Thus, $\rank(M_{P_k}) = n \cdot \rank(M_{g_{ik}^B}) 
 = n^{d-2}$.
 
 By induction hypothesis, we also know that the set of non-zero columns in $M_{g_{ik}^B}$ and 
 $M_{g_{i,k'}^B}$ are disjoint for $k \neq k'$. Thus, for $k \neq k'$, $M_{P_k}$ and $M_{P_{k'}}$ 
 do not contain any common non-zero column. Hence, 
 \[ \rank(M_{f_{ij}^B}) = \sum\limits_{k \in [n]}\rank(M_{P_k}) = n^{d-1} ~. \]
 To prove statement $2$, it is sufficient to observe that each non-zero 
 entry in $M_{f_{ij}^B}$ is present in a column such that the monomial indexing it is divisible by 
 some variable in the $j^{th}$ column of $A^d$ and is not divisible by any variable present in 
 other columns of $A^d$. Thus, for $j \neq j'$, the set of non-zero columns in $M_{f_{ij}^B}$ and 
 $M_{f_{ij'}^B}$ are disjoint.
 \end{proof} 
 
 For the case when $d$ is a odd integer, let us denote the polynomial $\prod_{k \in [d-1]}A^k(i,j)$ 
 by $f_{ij}$ and $\prod_{k \in [d]}A^k(i,j)$ by $g_{ij}$ for all $i, j \in [n]$. Thus, $g_{ij} = 
 \sum_{k \in [n]} f_{ik}x^{d}_{kj}$ which implies $g_{ij}^B = \sum_{k \in [n]} f_{ik}^By^{d}_{kj}$.
 Thus, $M_{g_{ij}^B}$ contains a copy of $M_{f_{ik}^B}$ for each $k$ and none of these copies have 
 a common non-zero row due to multiplication by $y^{kj}$. Moreover, we know that $M_{f_{ik}^B}$ and
 $M_{f_{ik'}^B}$ do not have any common non-zero columns.
 Thus, $M_{g_{ij}^B}$ contains a copy of the full rank sub-matrix of $M_{f_{ik}^B}$ for each $k$ with no common non-zero rows or columns. Hence, 
 $ \rank(M_{g_{ij}^B}) = \sum\limits_{k \in [n]} \rank(M_{f_{ik}^B}) = n^{d-1} ~.$
  
 \section{Complete Proof of Lemma \ref{l:2}}
 \label{app:repeat-proof}
 \begin{lemma-a}
 Let $\Phi$ be a $(s,d)$-product-sparse formula over the 
 set of variables $\{y_1, \dots, y_m\}$ and $\{z_1, \dots, z_m\}$, and 
 let $v$ be a node in $\Phi$. Denote the product-sparse depth of $v$ by $d(v)$. If $v$ is $k$-weak, then,
 $\maxrank(M_v) \le 2^{s \cdot d(v)} \cdot |\Phi_v| \cdot 2^{b(v)-k/2} ~.$ 
 \end{lemma-a}
 \begin{proof}
 We will prove by induction on $|\Phi_v|$. \\
 {\bf Base Case:} $v$ is a leaf node. By definition, the polynomial 
 computed at node $v$ is either a constant from the field or a 
 single variable. Thus, $\maxrank(M_v) \le 1$. Thus the lemma follows. \\
 {\bf Inductive Step:} Let $v$ be a node in $\Phi$ and assume that the 
 lemma holds for all nodes $u$ in $\Phi$ such that $|\Phi_u| < |\Phi_v|$. 
 We consider the following cases:
 \begin{enumerate}
  \item $v$ is a $k$-unbalanced node. 
  
  Thus, $a(v) \le b(v) - k$. 
  Hence, using Proposition \ref{p:2}, $\maxrank(M_v) \le 2^{b(v) - k} 
  \le 2^{sd(v)} \cdot |\Phi_v| \cdot 2^{b(v) - k/2}$. In the rest of the cases, we can 
  assume that $v$ is not $k$-unbalanced.
  
  \item $v$ is a disjoint product gate with children $v_1$ and $v_2$ and $v$.
  
  Thus, we have 
  $Y_{v_1} \cap Y_{v_2} = \phi$ and $Z_{v_1} \cap Z_{v_2} = \phi$. Thus, 
  $b(v) = b(v_1) + b(v_2)$. Without loss of generality, we can assume 
  that $b(v) \le 2b(v_1)$. Since $v$ is not $k$-unbalanced, every central 
  path that reaches $v_1$ must be $k$-unbalanced as otherwise we can 
  extend such a path to $v$ that will remain central and not $k$-unbalanced. 
  Thus $v_1$ is also $k$-weak and we have $|\Phi_{v_1}| < 
  |\Phi_v|$. Hence by induction hypothesis, 
  $ \maxrank(M_{v_1}) \le 2^{sd(v_1)} \cdot |\Phi_{v_1}| \cdot 2^{b(v_1) - k/2} ~. $
  Using Proposition \ref{p:2}, we have,
  $ \maxrank(M_{v_2}) \le 2^{a(v_2)} \le 2^{b(v_2)} ~.$
  Hence by Proposition \ref{p:4}, 
  \begin{eqnarray*}
  \maxrank(M_v) &=& \maxrank(M_{v_1}) \cdot \maxrank(M_{v_2}) \\ 
     & \le & 2^{sd(v_1)} \cdot |\Phi_{v_1}| \cdot 2^{b(v_1) + b(v_2)- k/2} \\
     & \le & 2^{sd(v)} \cdot |\Phi_{v}| \cdot 2^{b(v) - k/2} ~.
  \end{eqnarray*}
  \item $v$ is a $s$-sparse product gate with children $v_1$ and $v_2$ and $v$. 
  Without loss of generality it can be assumed that $v$ is not disjoint.
  
  Let us suppose that the product-sparse depth 
  of $v$ is $d$. Without loss of generality, assume that $v_2$ computes 
  a sparse polynomial having at most $2^s$ number of monomials. Since 
  $v$ is not disjoint, product-sparse depth of $v_1$ is at most 
  $d - 1$. Thus using Corollary \ref{c:3},
  \begin{eqnarray} \label{eqn:2}
  \maxrank(M_v) \le 2^s \cdot \maxrank(M_{v_1}) 
  \end{eqnarray}
  Consider the following cases based on whether $b(v) \le 2b(v_1)$ or 
  not.
  
  If $b(v) \le 2b(v_1)$, then $v_1$ is also $k$-weak. Therefore, by 
  induction hypothesis, 
  \[ \maxrank(M_{v_1}) \le 2^{s(d-1)} \cdot |\Phi_{v_1}| \cdot 
     2^{b(v_1) - k/2} \le 2^{s(d-1)} \cdot |\Phi_{v}| \cdot 
     2^{b(v) - k/2} ~.\]
  Thus, using Equation \ref{eqn:2},
  $ \maxrank(M_v) \le 2^{sd} \cdot |\Phi_v| \cdot 2^{b(v) - k/2} ~. $
  If $b(v) > 2b(v_1)$, then $b(v_1) < b(v)/2 < b(v) - k/2$ since $b(v) 
  \ge k$. Therefore using Proposition \ref{p:2}, $\maxrank(M_{v_1}) 
  \le 2^{a(v_1)} \le 2^{b(v_1)} < 2^{b(v) - k/2}$. Therefore, 
  $\maxrank(M_v) \le 2^s \cdot 2^{b(v) - k/2} \le 2^{sd} \cdot |\Phi_v| 
  \cdot 2^{b(v) - k/2} ~.$
  \item $v$ is a plus gate with children $v_1$ and $v_2$. 
  
  We know that $b(v) \le b(v_1) + b(v_2)$. Without loss of generality, 
  assume that $b(v) \le 2b(v_1)$ which implies that $v_1$ is also 
  $k$-weak. Hence by induction hypothesis,
  \[ \maxrank(M_{v_1}) \le 2^{sd(v_1)} \cdot |\Phi_{v_1}| \cdot 2^{b(v_1) - k/2} \le 
     2^{sd(v)} \cdot |\Phi_{v_1}| \cdot 2^{b(v) - k/2} ~. \]
  We consider the following cases based on whether $b(v) \le 2b(v_2)$ 
  or not:
  
  If $b(v) \le 2b(v_2)$, then $v_2$ is also $k$-weak. Hence by 
  induction hypothesis, 
  \[ \maxrank(M_{v_2}) \le 2^{sd(v_2)} \cdot |\Phi_{v_2}| \cdot 2^{b(v_2) - k/2} \le 
     2^{sd(v)} \cdot |\Phi_{v_2}| \cdot 2^{b(v) - k/2}. \]
  Hence by Proposition \ref{p:3},
  \begin{eqnarray*}
  \maxrank(M_v) & \le & \maxrank(M_{v_1}) + \maxrank(M_{v_2}) \\
     & \le & 2^{sd(v)} \cdot (|\Phi_{v_1}| + |\Phi_{v_2}|) \cdot 2^{b(v) - k/2} \\ 
     & \le & 2^{sd(v)} \cdot |\Phi_v| \cdot 2^{b(v) - k/2} ~.
  \end{eqnarray*}
  If $b(v) > 2b(v_2)$, then $b(v_2) < b(v)/2 < b(v) - k/2$ since 
  $b(v) \ge k$. Hence by proposition \ref{p:2}, 
  $ \maxrank(M_{v_2}) \le 2^{a(v_2)} \le 2^{b(v_2)} < 2^{b(v) - k/2} < 2^{sd(v)} \cdot 2^{b(v) - k/2} ~.$
  Hence by Proposition \ref{p:3}, 
  \[ \maxrank(M_v) \le 2^{sd(v)} \cdot (|\Phi_{v_1}| + 1) \cdot 2^{b(v) - k/2} \le 
     2^{sd(v)} \cdot |\Phi_v| \cdot 2^{b(v) - k/2}. \]
 \end{enumerate}
 \end{proof}

%% file: main.bbl
\begin{thebibliography}{10}

\bibitem{ASSS12}
Manindra Agrawal, Chandan Saha, Ramprasad Saptharishi, and Nitin Saxena.
\newblock Jacobian {H}its {C}ircuits: {H}itting-sets, {L}ower {B}ounds for
  {D}epth-d {O}ccur-k {F}ormulas {\&} {D}epth-3 {T}ranscendence {D}egree-k
  {C}ircuits.
\newblock In {\em Proceedings of the 44th ACM symposium on Theory of
  computing}, pages 599--614, 2012.

\bibitem{AV08}
Manindra Agrawal and V.~Vinay.
\newblock Arithmetic {C}ircuits: {A} {C}hasm at {D}epth {F}our.
\newblock In {\em Proceedings of the 2008 49th Annual IEEE Symposium on
  Foundations of Computer Science}, FOCS '08, pages 67--75, 2008.

\bibitem{DMPY11}
Zeev Dvir, Guillaume Malod, Sylvain Perifel, and Amir Yehudayoff.
\newblock Separating {M}ultilinear {B}ranching {P}rograms and {F}ormulas.
\newblock In {\em Proceedings of the 44th Symposium on Theory of Computing
  Conference, STOC 2012}, pages 615--624, 2012.

\bibitem{GR98}
D.~Grigoriev and A.~Razborov.
\newblock Exponential {C}omplexity {L}ower {B}ounds for {D}epth 3 {A}rithmetic
  {C}ircuits in {A}lgebras of {F}unctions over {F}inite {F}ields.
\newblock In {\em Proceedings of the 39th Annual Symposium on Foundations of
  Computer Science}, FOCS '98, pages 269--278, 1998.

\bibitem{GK98}
Dima Grigoriev and Marek Karpinski.
\newblock An {E}xponential {L}ower {B}ound for {D}epth 3 {A}rithmetic
  {C}ircuits.
\newblock In {\em Proceedings of the Thirtieth Annual ACM Symposium on the
  Theory of Computing (STOC)}, pages 577--582, 1998.

\bibitem{GKKS12}
Ankit Gupta, Pritish Kamath, Neeraj Kayal, and Ramprasad Saptharishi.
\newblock An {E}xponential {L}ower {B}ound for {H}omogeneous {D}epth {F}our
  {A}rithmetic {C}ircuits with {B}ounded {B}ottom {F}anin.
\newblock {\em Electronic Colloquium on Computational Complexity (ECCC)},
  19:98, 2012.

\bibitem{J08}
Maurice~J. Jansen.
\newblock Lower {B}ounds for {S}yntactically {M}ultilinear {A}lgebraic
  {B}ranching {P}rograms.
\newblock In {\em Proceedings of the 33rd international symposium on
  Mathematical Foundations of Computer Science}, MFCS '08, pages 407--418,
  2008.

\bibitem{Koi10}
Pascal Koiran.
\newblock Arithmetic {C}ircuits: {T}he {C}hasm at {D}epth {F}our {G}ets
  {W}ider.
\newblock {\em Computing Research Repository}, abs/1006.4700, 2010.

\bibitem{NW95}
N.~Nisan and A.~Wigderson.
\newblock Lower {B}ounds on {A}rithmetic {C}ircuits via {P}artial
  {D}erivatives.
\newblock In {\em Proceedings of the 36th Annual Symposium on Foundations of
  Computer Science}, FOCS '95, pages 16--25, 1995.

\bibitem{RY08b}
R.~Raz and A.~Yehudayoff.
\newblock Lower {B}ounds and {S}eparations for {C}onstant {D}epth {M}ultilinear
  {C}ircuits.
\newblock In {\em Computational Complexity, 2008. CCC '08. 23rd Annual IEEE
  Conference on}, pages 128--139, june 2008.

\bibitem{Raz06}
Ran Raz.
\newblock Separation of {M}ultilinear {C}ircuit and {F}ormula {S}ize.
\newblock {\em Theory of Computing}, 2(1):121--135, 2006.

\bibitem{Raz09}
Ran Raz.
\newblock Multi-linear {F}ormulas for {P}ermanent and {D}eterminant are of
  {S}uper-polynomial {S}ize.
\newblock {\em Journal of ACM}, 56:8:1--8:17, April 2009.

\bibitem{RSY08}
Ran Raz, Amir Shpilka, and Amir Yehudayoff.
\newblock A {L}ower {B}ound for the {S}ize of {S}yntactically {M}ultilinear
  {A}rithmetic {C}ircuits.
\newblock {\em SIAM Journal of Computation}, 38(4):1624--1647, 2008.

\bibitem{S07}
Nitin Saxena.
\newblock Diagonal {C}ircuit {I}dentity {T}esting and {L}ower {B}ounds.
\newblock {\em Electronic Colloquium on Computational Complexity (ECCC)},
  14(124), 2007.

\bibitem{Shp01}
Amir Shpilka.
\newblock Affine {P}rojections of {S}ymmetric {P}olynomials.
\newblock In {\em Proceedings of the 16th Annual Conference on Computational
  Complexity}, CCC '01, pages 160--, Washington, DC, USA, 2001. IEEE Computer
  Society.

\bibitem{SW01}
Amir Shpilka and Avi Wigderson.
\newblock Depth-3 {A}rithmetic {C}ircuits over {F}ields of {C}haracteristic
  {Z}ero.
\newblock {\em Computational Complexity}, 10(1):1--27, 2001.

\bibitem{SY10}
Amir Shpilka and Amir Yehudayoff.
\newblock Arithmetic {C}ircuits: {A} {S}urvey of {R}ecent {R}esults and {O}pen
  {Q}uestions.
\newblock {\em Foundations and Trends in Theoretical Computer Science},
  5(3-4):207--388, March 2010.

\end{thebibliography}
